\documentclass[12pt,twoside]{article}
\usepackage[left=2.5cm,top=2cm,right=2.5cm,bottom=2.5cm,bindingoffset=0cm]{geometry}
\usepackage{amssymb,latexsym}
\usepackage{mathtools}
\usepackage{graphicx}
\usepackage{appendix}
\usepackage{fancyhdr}
\usepackage{cite}
\usepackage[english]{babel}
\usepackage{nicefrac}
\usepackage{bbm}
\usepackage{amssymb}
\usepackage{tabularx}
\usepackage{epsfig}
\usepackage{amsfonts}
\usepackage{amscd}
\usepackage{amsmath,amsthm}
\usepackage{lmodern}
\usepackage{enumerate}
\usepackage{mathrsfs}
\usepackage[linktocpage,colorlinks=true,citecolor=blue,linkcolor=blue,urlcolor=blue]{hyperref}
\newcommand{\arxiv}[1]{\href{http://arxiv.org/abs/#1}{\texttt{arXiv:#1}}}

\theoremstyle{plain}

\newtheorem{theorem}{Theorem}[section]
\newtheorem{lemma}[theorem]{Lemma}

\newtheorem{proposition}[theorem]{Proposition}
\newtheorem{fact}[theorem]{Fact}

\theoremstyle{definition}

\theoremstyle{remark}
\newtheorem*{remark}{Remark}

\newcommand{\be}{\begin{equation}}
\newcommand{\ee}{\end{equation}}
\newcommand{\bea}{\begin{eqnarray}}
\newcommand{\eea}{\end{eqnarray}}
\newcommand{\bfa}{\begin{fact}}
	\newcommand{\efa}{\end{fact}}
\newcommand{\bin}{\begin{inequality}}
	\newcommand{\ein}{\end{inequality}}
\def \l {{\lambda}}
\def \n {\bar{n}}
\def \k {{\kappa}}
\def \r {{\rho}}
\def \p {{\mathbb{P}}}

\title{\bf On the combinatorics of exclusion \\in Haldane fractional statistics}

\author{\textsc{Nour-Eddine Fahssi}\\
	\small \textit{Lab. High Energy Physics, Modeling and Simulation}, \\ [-0.3em]
	\small \textit{FS, Mohammed V University of Rabat.}\\[-0.3em]
	\small and\\ [-0.3em]
	\small \textit{Lab. M2CAN, Dept of Mathematics, Hassan Second }\\ [-0.3em] 
	\small \textit{University of Casablanca, FST Mohammedia, Morocco.\thanks{\emph{Permanent address.}}} \\
	\small \href{mailto:n.fahssi@live.fr}{\tt n.fahssi@live.fr}}

\date{}
\begin{document}
	\maketitle
\begin{abstract}
This paper is a revision of the combinatorics of fractional exclusion statistics (FES). More specifically, the following exact statement of the generalized Pauli principle is derived: for an $N$-particles system exhibiting FES of extended parameter \mbox{$g=q/r$} ($q$ and $r$ are co-prime integers such that $0 < q \leq r$), we found that the allowed occupation number of a state is smaller than or equal to $r-q+1$ and \emph{not} to $1/g$ whenever $q\neq 1$ and, moreover, the global occupancy shape (merely represented by a partition of $N$) is admissible if the number of states occupied by at least two particles is less than or equal to $(N-1)/r$ ($N \equiv 1 \pmod r$). These counting rules allow distinguishing infinitely many families of FES systems depending on the parameter $g$ and the size $N$. As an application of the main result, we study the probability distributions of occupancy configurations. For instance, the number of occupied states is found to be a hypergeometric random variable. Closed-form expressions for the expectation values and variances in the thermodynamic limit are presented. By way of comparison, we obtain parallel results regarding the Gentile intermediate statistics and demonstrate subtle similarities and contrasts with respect to FES.
\\

\hfill{PACS number(s): 05.30.Pr, 02.10.Ox}
\end{abstract}\newpage
\tableofcontents

\section{Introduction}\label{sec:introduction}

Fractional exclusion statistics (FES) is an archetype of unconventional statistics. Since it was introduced by Haldane (1991) to explain the properties of quasi-particles in the fractional quantum Hall effect~\cite{Hald}, FES has been a subject of intense research and has found applications in numerous models of interacting particles. Nowadays, the literature on the topic is voluminous; we refer, e.g., to papers~\cite{ref2,ref3,ref4,ref5,ref6,anghel,anghel2,CS,Hu} and references cited therein.

Generally, a FES system consists of a countable number of species of particles; each species consists of a finite number of single-particle states. Haldane's proposal is based on a generalization of the Pauli principle. Explicitly, in the case with only one species, an $N$th (quasi-)particle added to a system of identical particles can occupy $d_N=K-g(N-1)$ single-particle states, where $K$ is the number of available states when $N = 1$ and the constant $g$ is a parameter of the ``statistical interaction''. The number $d_{N}$ represents the dimension of the one-particle Hilbert space obtained by keeping the quantum numbers of the $N-1$ other particles fixed. Naturally, the conventional Bose-Einstein (BE) and Fermi-Dirac (FD) statistics are recovered for $g = 0$ (no exclusion) and $g = 1$ (perfect Pauli exclusion), respectively. In these notes, FES with parameter $g$ will be referred to as FES$_g$.

The total size of the full Hilbert space of many-particle states for FES systems is postulated to be~\cite{Hald, ref3} \be \label{HW} W_{g}(K,N)= \binom{d_{N}+N-1}{N},\ee where ${a \choose b}= a!/(b! (a-b)!)$ is a binomial coefficient. As mentioned by Wu, the statistical weight~\eqref{HW} is a generalization of Yang-Yang state counting~\cite{yang}. The thermodynamic properties of FES gazes were widely studied, primarily by Wu~\cite{ref3} and Isakov~\cite{ref4}. For instance, the average occupation number is found to be
\[ \n_g(\epsilon)=\dfrac{1}{f(\xi)+g}  < \frac{1}{g}, \] where $\xi =e^{\beta (\epsilon-\mu)}$, $\epsilon$ is the single particle energy, $\beta$ the inverse temperature, $\mu$ the chemical potential of the system and the function $f(\xi)$ satisfies the functional equation  $f^g (1+f)^{1-g}= \xi$.

Clearly, to have a combinatorial meaning, the number of particles $N$ has to be congruent to $1 \pmod r$ so that the dimension $d_N$, and accordingly $W_g(K,N)$, is a whole number. Thus, if $N = r P +1$ for some integer $P$, then $d_{N+r}-d_N=-q$, viz. adding $r$ particles reduces the number of available states by $q$. The number of quantum states~\eqref{HW} takes now the form \be \label{HWbis} W_g(K,0)=1, \quad \hbox{and}  \quad W_g(K,N) ={K+(r-q)P \choose r P+1}. \tag{$1'$}\ee Note that $W_g(K,N)=0$ if $P > (K-1)/q$.

In Ref.~\cite{Poly}, Polychronakos proposed an extensive model which accurately gives back the statistical mechanics of FES in the thermodynamic limit. Extensivity (or multiplicativity) here means that, for large $K$, the grand partition function is the $K$th power of a $K$-independent function~\cite{Poly}. However, the price paid for this microscopic realization is the occurrence of negative probabilities; see also~\cite{Nayak}. Now, it is understood that this problem occurs because Haldane statistics is not extensive and, unlike the Pauli principle, the exclusion operates on sets of levels~\cite{ref6}. Chaturvedi and Srinivasan~\cite{chat}, and subsequently Murthy and Shankar~\cite{neg}, showed how negative weights may be avoided for $g=1/2$ (semions) and for $g=1/3$, and indicated -- without being explicit -- that ``there is an algorithm to derive  single-particle occupation probabilities for arbitrary $g=1/m$ though this gets complicated for larger $m$''~\cite{neg}. In this letter, we revisit and solve this problem in a closed form when the parameter $g$ is generally any irreducible fraction: $g=q/r$, where $q$ and $r$ are coprime and $0 < q \leq r$. Moreover, while doing this, we revise and generalize the exclusion rules of FES. Our approach is purely combinatorial; it leads to the following exclusion principle: \emph{An occupancy configuration is allowed if} (1) \emph{the maximal number of particles that each state can accommodate is $r-q+1$, and \emph{not} to $g^{-1}$ whenever $q\neq 1$}, and (2) \emph{the configurations in which the number of states occupied by two or more particles is greater than $(N-1)/r$ are forbidden}. This allows us to distinguish infinitely many families of FES$_g$ systems depending on $g$ and $N$.

In Section~\ref{s2}, we state our main result (Theorem~\ref{HWweight}) and interpret its combinatorial consequences. Section~\ref{s3} deals with an application to the statistics of occupancy configurations. By way of comparison with other exotic models, we derive similar results for the Gentile intermediate statistics. We end with some concluding remarks in Section~\ref{s4}.
\section{The Exact combinatorics}\label{s2}
In order to state our main result, we need some background on the theory of partitions. A \emph{partition} of a non-negative integer $N$ is a non-increasing sequence of positive integers whose sum is $N$. To indicate that $\l$ is a partition of $N$, we write $\l \vdash N$ and denote \mbox{$\l=(1^{k_1} 2^{k_2} \ldots N^{k_N})$}, where $\sum_{i=1}^N i k_i = N$ and $k_i$ designates the multiplicity of the part $i$; the sum $\ell(\l)=\sum_{i=1}^N k_i$ is called the \emph{length} of $\l$. The \emph{Ferrers diagram} of $\l$ is a pattern of dots, with the $j$th row having the same number of dots as the $j$th term in $\l$.

Suppose we have $N$ indistinguishable balls (particles) randomly distributed into $K$ labeled boxes (states). An \emph{occupancy configuration} is said to be of shape $\l=(1^{k_1} 2^{k_2} \ldots N^{k_N}) \vdash N$ if $k_i$ states are occupied by $i$ particles ($i=1,\ldots, N$) and the number of non-vacant states $\ell(\l)$ is less than or equal to $K$. Moreover, if no parts of $\l$ exceed a fixed integer $m$, the corresponding configuration is additionally characterized by $\ell(\l^*) \leq m$, where $\l^*$ stands for the \emph{conjugate} partition  of $\l$, that is, the partition whose Ferrers diagram is obtained from $\l$ by reflection with respect to the diagonal so that rows become columns and columns become rows.
\subsection{The main result}
The combinatorics of FES is encoded in the following result.
\begin{theorem} \label{HWweight} For $g=q/r$ and $N \equiv 1 \pmod{r}$ , the number of microstates \emph{\eqref{HWbis}} can be written as
\be \label{Wbis}
W_g(K,N)=\sum_{{\l \, \vdash N}}w_g(\l) \frac{\ell(\l)!}{k_1! \, k_2! \cdots k_N!} \binom{K}{\ell(\l)}, \ee  where the sum runs over partitions of $N$, and \be  w_g(\l) = \label{HWw} \frac{ \displaystyle \binom{(N-1)/r}{\ell(\l)-k_1}}{ \displaystyle \binom{\ell(\l)}{k_1}}  \prod_{j=0}^{r-q} \binom{r-q}{j}^{k_{j+1}} H[r-q+1 - \ell(\l^*)]; \ee the function $H$ being the Heaviside step function $H[0]=1)$. 
\end{theorem} \noindent For the sake of readability, we report the proof in Subsection~\ref{proof}.

Displayed in the form~\eqref{Wbis}, $W_g(K,N)$ may be interpreted as follows. A  configuration $\l$ being fixed, the factor $\frac{\ell(\l)!}{k_1! \, k_2! \cdots k_N!} \binom{K}{\ell(\l)}$  counts the ways to choose $\ell(\l)$ non-vacant states out of $K$ ones and arrange $k_i$ states with $i$ particles ($i=1,\ldots ,N$) among them. The result is then weighted by a configuration-dependent function $w_g(\l)$. Due to the expression~\eqref{HWw}, the sum in Eq.~\eqref{Wbis} runs actually over restricted partitions of $N$.

In the case with $g=1/2$, the weight~\eqref{HWw} reads, for $\l=(1^{k_1}2^{k_2}) \vdash N$, as
\[w_{1/2}(\l)=\binom{(N-1)/2}{k_2}\binom{k_1+k_2}{k_1}^{-1},\] which is exactly the formula derived by Chaturvedi and Srinivasan in their microscopic interpretation of semion statistics~\cite{chat}.

Obviously, the one-configuration weight $w_g$ characterizes the studied occupancy model. In fact, the form~\eqref{Wbis} is generic to any statistics based on ``Balls-in-Boxes'' models \emph{with distinguishable boxes}. For instance, the number of microstates for the Gentile intermediate statistics (GS)~\cite{gent} can be cast in the form~\eqref{Wbis}. Indeed, it is well known that the partition function is~\cite{Poly,fahssi} \be \label{Z} \mathcal{Z}(z)= \sum_{N=0}^{\infty} W_G(K,N) z^N=\left(1+z+\ldots + z^G\right)^K,\ee  where $z$ is the fugacity and $G$ is the order of Gentile statistics~\footnote{Throughout this letter, we set GS$_G$ to designate GS of oreder $G$}. Using the multinomial theorem to expand the power in~\eqref{Z} and extracting the coefficient of $z^N$, we obtain the identity
\be \label{G} W_G(K,N)= \sum_{\{k_i\}} \frac{K!}{k_1! \cdots k_G!(K-k_1- \cdots -k_G)!},\ee where the sum runs over all $G$-tuples $(k_1, \ldots , k_G)$ subject to $k_1+2k_2 + \ldots + G k_G =N $, i.e. over restricted partitions of $N$. Thus $W_G(K,N)$ can be written as~\eqref{Wbis} with a weight given by:
\be \label{wG} w_G(\l) = H[G-\ell(\l^*)]. \ee In the table below, we summarize our calculations of the weight $w(\l)$ for the most known statistics; see also~\cite{fahssi}. The so-called  $\gamma$-statistics, introduced as an ansatz in~\cite{AS,Poly}, interpolates between FD ($\gamma=1$), BE ($\gamma=-1$) and the classical Maxwell-Boltzmann (MB) statistics ($\gamma=0$). 
{\small \begin{center}
	\begin{tabular}{ccccc}
		\hline
 Statistics &&&&$w(\l)$ \\
  \hline\hline
  BE &&&& 1 \\
  FD &&&& $1$ if $\l = (1^N)$, 0 otherwise \\
  MB &&&& $\displaystyle (1!^{k_1} 2!^{k_2}\cdots N!^{k_N})^{-1}$ \\
   FES$_{q/r}$ &&&& Eq.~\eqref{HWw} \\
  GS$_G$ &&&& Eq.~\eqref{wG} \\
  $\gamma$-statistics &&&& $ \gamma^N \binom{\gamma^{-1}}{1}^{k_1} \binom{\gamma^{-1}}{2}^{k_2} \cdots \binom{\gamma^{-1}}{N}^{k_N}$ \\ \hline
\end{tabular}\end{center}}
\subsection{Interpretation of the weight $w_g$}
From the expression of $w_g$, we underline the following features:   \begin{enumerate}
	\item[(1)] the weights $w_g (\l)$ are fractional and non-negative definite,
	\item[(2)] the weights $w_g(\l)$ depend only upon $P \coloneqq (N-1)/r$ and the difference $r-q$,
	\item[(3)] the allowed occupation number for a single-state does not exceed $r-q+1$ and \emph{not} $1/g$ whenever $q \neq 1$. We recall, however, that the average occupation number $\n_g(\epsilon)$ does not exceed $1/g \leq r-q+1$,
	\item[(4)] Since the binomial coefficient $\binom{P}{\ell(\l)-k_1}$ in~\eqref{HWw} vanishes if $P < \sum_{i=2}^m k_i$ , the corresponding configuration does not contribute to the total weight.
\end{enumerate} The last observation is crucial. It stipulates that a necessary condition for permissible configurations is that the number of states occupied by two particles or more is less than or equal to $P$, that is, the Ferrers diagram of $(2^{k_{2}} \ldots (r-q+1)^{k_{r-q+1}})$, extracted from $\l$, fits inside the rectangle $[P \times (r-q+1)]$.

Let us incorporate the above-formulated rules as follows: 
\medskip

\noindent \textbf{Generalized Exclusion Principle}.
	\emph{A configuration of shape $ \l \vdash n$ is admissible if and only if the following constraints are fulfilled:}
\begin{enumerate}[leftmargin=5cm]
	\item[$C_1$\; \emph{:}] \quad $\ell(\l) \leq K$ \; \quad (\emph{by definition}),
	\item[$C_2$\; \emph{:}] \quad  $\ell(\l^*) \leq r-q+1$  \quad (\emph{at most $r-q+1$ particles per state}),
	\item[$C_3$\; \emph{:}]  \quad $\displaystyle \sum_{i=2}^m k_i \leq \frac{N-1}{r} \leq  \frac{K-1}{q}$ \quad ($w_g(\l) \neq 0$ \emph{and} $d_N \geq 1$). \end{enumerate}

Therefore, the exclusion operates not only on the ``microscopic'' level (condition $C_2$), but also on the ``macroscopic'' level (condition $C_3$). To illustrate, we implement this in two specific examples:

 $\bullet$ Let $g=1/3$. Here the maximal allowed occupancy of a state is $3$ and  \[w_{1/3}(\l)= \binom{(N-1)/3}{k_2+k_3}\binom{k_1+k_2+k_3}{k_1}^{-1} 2^{k_2}, \] for $\l=(1^{k_1}2^{k_2}3^{k_3}) \vdash N$. This formula was obtained by Murthy and Shankar using an exactly solvable model~\cite{neg}. For an example, take, say, $N=10$. By the constraint $C_2$, 14 configurations may contribute (depending on $K \geq 4$), among which the configurations $(1^2 2^4)$, $(2^5)$, $(1 2^3 3)$ and $(2^2 3^2)$  are forbidden by the constraint $C_3$:
\begin{center}\qquad \quad  \includegraphics[width=6cm]{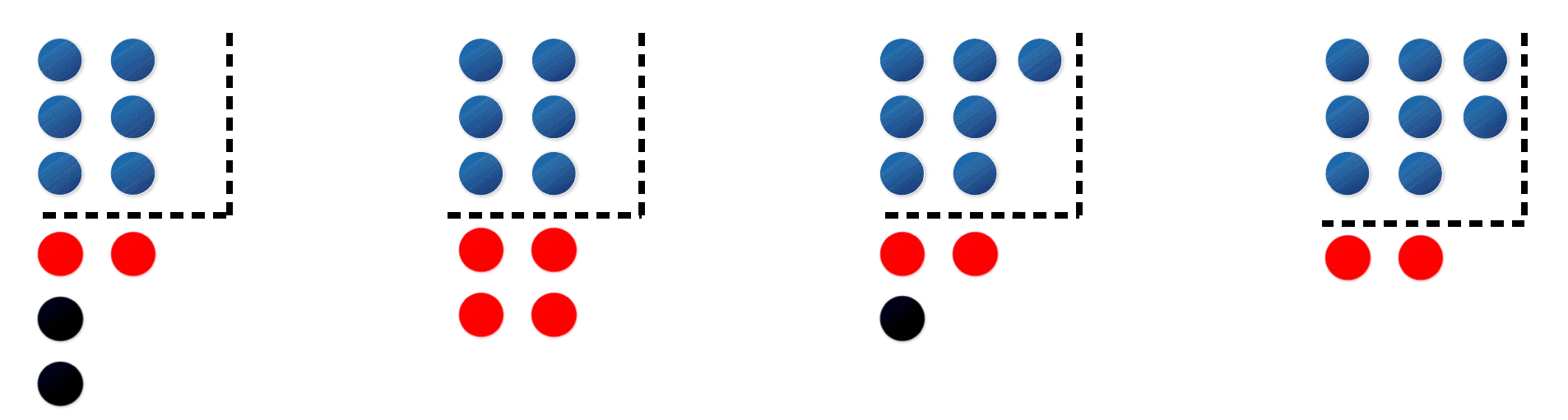}\end{center}

$\bullet$ Let $g=3/5$ and $N=16$. Here the maximal allowed occupancy is again $3$. Among the 231 partitions of 16, only 10 may contribute to the total weight: $(1^{16})$, $(1^{14} 2)$, $(1^{12} 2^2)$, $(1^{13} 3)$, $(1^{10} 2^3)$, $(1^{11} 2\,3)$, $(1^9 2^2 3)$, $(1^{10} 3^2)$, $(1^8 2\,3^2)$, $(1^7 3^3)$, each of which contributes only if its length is less than or equal to $K \geq 10$.

\begin{proposition} \label{prop3}
	For $N \leq K$, the number of permissible configurations is given by: \be \label{conf} \binom{(N-1)/r +r-q}{r-q}. \ee
\end{proposition} 
\begin{proof} Clearly, when $N \leq K$ the condition $C_1$ and the inequality in the right of the constraint $C_3$ are satisfied. Thus, a configuration $\l$ is likely if and only if the inequality in the left of the condition $C_3$ holds true. Therefore, the number of allowed configurations is the number of solutions of $k_2 + k_3 + \cdots + k_m \leq (N-1)/r$ in nonnegative integers. The result follows from the known fact that the number of solutions of $x_1 + x_2 + \cdots + x_k \leq p$ is given by $\binom{p+k}{k}$ (cf.~\cite[p.103]{vanLint}). \end{proof}
By way of comparison, the exact exclusion rules for GS$_G$ are, in addition to $C_1$, $\ell(\l^*) \leq G$ and $N \leq G K$. Thus, the number of permitted configurations is simply that of the partitions of $N$ with no more than $K$ parts; no part exceeding $G$. This number is the coefficient of $q^N$ in the Gaussian polynomial $\left[ \!{K+G \atop K} \!\right]_q$~\cite[Chap.3]{Andr}. When $N\leq K$, this reduces to the number of partitions with largest part not exceeding $G$. We also emphasize that if $G=r-q+1$, then $W_G(K,N)$  majorizes $W_g(K,N)$ since the exclusion principle of FES is more restrictive.

It is worth noting that, in view of the constraints $C_2$ and $C_3$, we can distinguish infinitely many families of FES systems according to $P=(N-1)/r$ and the difference $ r-q $. Indeed, representing an $N$-particle system fulfilling FES$_g$ by the pair $(N,g=q/r)$, two systems $(N,g=q/r)$ and $(N',g')$ are subject to the same exclusion rules if there exist an integer $j >0$ not a multiple of $r-q$ such that
\be   g'=\frac{j}{r-q+j}, \quad \hbox{and} \quad
\frac{N'-1}{r-q+j}=\frac{N-1}{r}.\ee
The semions, for example, belong to the family with $g=j/(j+1)$, the \emph{semionic family}. Clearly, the Bose and Fermi statistics are recovered in the limits $j=0$ and $j \to \infty$ respectively.
\subsection{Proof of Theorem~\ref{HWweight}}
\label{proof}	
To prove Theorem~\ref{HWweight}, we need the following identity:
\begin{lemma} Let $P$, $n$ and $k$ be positive integers. Then
	\be \label{HF} \binom{k P}{n}= \sum_{\{l_i\}} \frac{P!}{l_1! \cdots l_k!(P-l_1- \cdots -l_k)!}
	\prod_{i=1}^k \binom{k}{i}^{l_i}, \ee where the sum runs over all $k$-tuples $(l_1, \ldots , l_k)$ subject to the constraint $l_1+2l_2 + \ldots +k l_k =n $. \end{lemma}

\begin{proof} We shall use the technique of generating function to prove the identity~\eqref{HF} (see~\cite{fahssi}). Let $t$ be an indeterminate. On one hand, we have by application of the binomial theorem  \be \label{ex1} (1+t)^{kP}=\sum_{n=0}^{kP}\binom{kP}{n} t^n,\ee and, on the other hand,  by the well-known multinomial theorem:  \bea \nonumber (1+t)^{kP}&=&\left((1+t)^k\right)^P = \left(\sum_{i=0}^k \binom{k}{i}t^i\right)^P = \sum_{l_0 + l_1 + \cdots l_k=P} \frac{P!}{l_0! \, l_1! \, \cdots l_k!}\prod_{i=0}^k \left( \binom{k}{i}t^i \right)^{l_i} \\
	\label{ex2} &=& \sum_{(l_1 , \ldots , l_k)} \left(\frac{P!}{ l_1! \, \cdots l_k! \, (P-l_1-l_2-\cdots l_k)! }\prod_{i=0}^k \binom{k}{i}^{l_i}\right) \; t^{l_1+2l_2+\cdots +k l_k}. \eea   The identity~\eqref{HF} follows by equating the coefficients of $t^n$ in the two expansions~\eqref{ex1} and~\eqref{ex2}.\end{proof}  \begin{proof}[Proof of Theorem~\ref{HWweight}]  Inserting the weight $w_g(\l)$, the RHS of \eqref{Wbis} can be displayed as  \[ \sum_{{\l \vdash N \atop \ell(\l^*) \leq r-q+1}} \left(
	\frac{P!}{k_2! \cdots k_{r-q+1}!(P-k_2-\cdots-k_{r-q+1})!} \prod_{i=1}^{r-q}
	{{r-q}\choose i}^{k_{i+1}} \right) {{K}\choose \ell(\l)}, \] where $P=(N-1)/r$. Taking into account that $\ell(\l)=\sum_{i=1}^{r-q}k_{i}=N-\sum_{i=1}^{r-q}ik_{i+1}$
	and putting $s=\sum_{i=1}^{r-q}ik_{i+1}$ (the integer $s$ ranges
	from 0 to $(r-q)P$ since $k_{r-q+1} \leq P$), we
	re-express the last formula as a double sum:
	\be \label{A2} \sum_{s=0}^{(r-q)P} \left( \sum_{\sum_{i=1}^{r-q}ik_{i+1}=s}\frac{P!}{k_2! \cdots k_{r-q+1}!(P-k_2-\cdots-k_{r-q+1})!} \prod_{i=1}^{r-q} {{r-q}\choose i}^{k_{i+1}}\right) {{K}\choose N-s}.\ee  Now we make the change of summation indices $l_i=k_{i+1}$ to write the inner sum as the RHS of formula~\eqref{HF}:
	\be \label{A3} \sum_{\sum_{i=1}^{r-q}i l_{i}=s}\frac{P!}{l_1! \cdots
		l_{r-q}!(P-l_1-\cdots-l_{r-q})!} \prod_{i=1}^{r-q} {{r-q}\choose i}^{l_{i}} = {{(r-q)P}\choose s}.
	\ee
	We deduce finally that the RHS of Eq.~\eqref{Wbis} reads \be \label{Vander} \sum_{s=0}^{(r-q)P}{{(r-q)P}\choose s}{{K}\choose N-s}={{K+(r-q)P}\choose N}=W_g(K,N), \ee where, to obtain the last equality, we employed the well-known Vandermonde's formula for binomial coefficients~\cite{vander}.  \end{proof} 
\begin{remark}
For $g>1$ ($r<q$), one may follow the proof above to check that $W_g(K,N)$ can as well be formally written in the form~\eqref{Wbis}, but the constraint of maximal occupancy became relaxed and the weights \emph{inevitably negative} for some configurations. Indeed, in this case, the weights are not positive definite since $\binom{r-q}{i}<0$ for odd $i$.\end{remark}
\section{The state-occupancy distributions in the thermodynamic limit} \label{s3} In Balls-in-Boxes models, a problem of interest is the statistics of occupation patterns, for example, the probability distributions of occupied/vacant cells or those accommodating a fixed number of balls, etc. In this section, we comparatively investigate these questions and more for FES and GS, and give a probabilistic application of our main result to the statistics of occupancies in the thermodynamic limit (i.e. $ N, K \to \infty $ and $ N / K $ is held bounded).
\subsection{Two probability measures} The combinatorial expression~\eqref{Wbis} suggest the following probability measures on the set of partitions of $N$, \be \label{prob} \p_\alpha(\l) = \frac{w_\alpha(\l )}{W_\alpha(K,N)}\, \frac{\ell(\l)!}{k_1! \, k_2! \cdots k_N!} \binom{K}{\ell(\l)},\ee the Haldane measure ($\alpha \equiv g$) and the Gentile measure ($\alpha \equiv G$). We regard $\p_\alpha (\l)$ as the probability of configuration $\l$. Actually, the probability so defined is the joint distribution of the random variables $k_2,k_3, \ldots,k_m$ ($m=r-q+1$ or $G$). For the semionic family and GS$_2$, we have \bea \label{probg}  \p_{g=1/2} (\l) &=& \dfrac{1}{ \binom{P+K}{N}} \binom{P}{k_2} \binom{K}{N-k_2}, \\ \label{probG} \p_{G=2} (\l) &=&  \frac{1}{W_{G=2}(K,N)} \binom{N-k_2}{k_2} \binom{K}{N-k_2}  , \eea respectively.

Interestingly, the distribution~\eqref{probg} shows that $k_2$ (or $\ell(\l)=N-k_2)$ is a usual \emph{hypergeometric} random variable of parameters $P$, $P+K$ and $N$, that is, $\p_{1/2}(\l)$ describes the probability of getting $k_2$ successes in $N$ draws without replacement where the sample population is $P+K$. Each draw is either success or failure and the population consists of exactly $P$ successes~\cite{feller}. More generally, we show that the number of occupied states follows a hypergeometric law with parameters $K$, $(r-q)P+K$ and $N$; see Eqs. \eqref{A2}, \eqref{A3} and \eqref{Vander}. On the other hand, the distribution~\eqref{probG} relative to GS$_{2}$ is unusual.

Consider now the average number of states with $i$ particles $$ \langle k_i \rangle_\alpha =  \sum_{\l \vdash n}k_i  \p_\alpha (\l), \quad (\alpha \equiv g,G ),$$ and set \be \kappa_i^{(\alpha)} \coloneqq \lim\limits_{{N,K \to \infty \atop N/K=r\r}} \frac{1}{K} \langle k_i \rangle_\alpha,\ee for the proportion of states accommodating $i$ particles. The variable $\r$ controls the thermodynamic limit. We also define the normalized variance:
\be \nu_i^{(\alpha)} \coloneqq \lim\limits_{{N,K \to \infty \atop N/K=r\r}}  \frac{1}{K}\left(\langle k_i ^2 \rangle_\alpha -\langle k_i \rangle_\alpha^2\right). \ee These limits exist for both GS and FES as we will show.

Using the expressions of the expectation value and the variance of hypergeometric distributions~\cite{feller}, we find the thermodynamic limit of the normalized mean and variance of the number of occupied states: 
\be \lim\limits_{{N,K \to \infty \atop N/K=r\r}} \frac{1}{K}\langle \ell(\l) \rangle_g = \frac{r \r}{1+(r-q) \r} ,\ee and \be \lim\limits_{{N,K \to \infty \atop N/K=r\r}} \frac{1}{K}\left(\langle \ell(\l) ^2 \rangle_g -\langle \ell(\l) \rangle_g^2 \right) = \frac{ r (r-q)(1-q \r )\r ^2}{(1+
	(r-q)\r)^3}, \ee where $\r = P/K$. Note that $\r$ ranges in the interval $(0,1/q)$ due to the constraint $C_3$. For the semion family $g=(j-1)/j$, we find
\bea \nonumber \kappa_1^{(g)}(\r) = \frac{j \r (1-\r)}{1+\r},&& \quad \kappa_2^{(g)} (\r) =\frac{j \r^2}{1+\r}, \\ \nonumber  \kappa_0^{(g)} (\r)= \frac{1-(j-1)\r}{1+\r},\qquad && \quad \nu_2^{(g)} (\r )= \frac{j(1-(j-1)\r ) \r ^2}{(1+\r)^3}.\eea
For $r-q\geq 2$, we do not have such explicit expressions. As for the Gentile statistics, we show the following
\begin{proposition} \label{Gent} In the thermodynamic limit, the mean number of states with $i$ particles is given by \be \label{asym} \kappa_i^{(G)} (\r) =  \frac{\big(1-x(\r)\big)x(\r)^i}{1-x(\r)^{G+1}},\ee for $i=0,1, \ldots , G$, where $\r=N/(GK)$ and $x(\r)$ is the (unique) positive solution of \be \label{eq} G \r = \frac{t \left(1-(G+1) \, t^G + G \, t^{G+1}\right)}{(1-t)
		\left(1-t^{G+1}\right)}.\ee Moreover, the following duality relation  holds true \be \label{dua} \kappa_i^{(G)}(1-\r)=\kappa_{G-i}^{(G)}(\r).\ee \end{proposition} 
\begin{proof} The mean number of states with $i$ particles under the Gentile measure is given by
		\bea \nonumber \langle k_i \rangle_G &=& \frac{1}{W_G(K,N)} \sum_{\sum_j jk_j=N} \frac{ k_i K!}{k_1! \cdots k_G!(K-k_1- \cdots -k_G)!}\\
		\nonumber	&=& \frac{1}{W_G(K,N)} \sum_{\sum_j jk_j=N}  \frac{ K (K-1)!}{k_1! \cdots (k_i-1)! \cdots k_G!(K-k_1- \cdots -k_G)!} \\
		\nonumber	&=& \frac{K}{W_G(K,N)} \sum_{\sum_j jr_j=N-i}  \frac{  (K-1)!}{r_1! \cdots r_i! \cdots r_G!(K-1-r_1- \cdots -r_G)!} \\
		&=& K \dfrac{W_G(K-1,N-i)}{W_G(K,N)}, \quad \hbox{for all $K$ and $N$.} \eea 
		To compute the thermodynamic limit, we use the following asymptotic estimate~\cite{fahssi} :  \[ W_G(K,N) \sim \frac{(1+x+x^2+ \cdots + x^G)^K}{x^{N+1}\sqrt{2\pi C K }} , \quad  \hbox{uniformly as $N,K \to \infty$, and $N/K$ finite,}\]  where $x=x(\r)$ is the positive real solution of Eq.~\eqref{eq} and $C$ is a constant depending on $x$ only. Whence,
		\[\k^{(G)}_i(\r)= \lim\limits_{{N,K \to \infty \atop N/K=G\r}}\dfrac{W_G(K-1,N-i)}{W_G(K,N)}= \frac{x(\r)^i}{1+x(\r)+x(\r)^2 + \cdots + x(\r)^G}= \frac{\big(1-x(\r)\big) x(\r)^i}{1-x(\r)^{G+1}}, \] as desired.
		
		The duality~\eqref{dua} follows immediately from the symmetry relation: $W_G(K,N)=W_G(K,GK-N)$~\cite{fahssi2}.\end{proof} Consequently, the normalized average number of occupied states under the Gentile measure is \[ \sum_{i=1}^G \kappa_{i}^{(G)}=\frac{x(\r) \left(x(\r)^G-1\right)}{x(\r)^{G+1}-1}.\]
	Closed explicit expressions are possible for $G \leq 4$. For instance,
\bea \nonumber \kappa_1^{(G=2)} (\r) &=& -\frac{1}{3}+
{1 \over 3} \sqrt{1+12 \r -12 \r ^2
}, \\ \nonumber \kappa_2^{(G=2)} (\r)&=& \frac{1}{6}+\r-\frac{1}{6} \sqrt{1+12 \r -12 \r ^2}\\
\nonumber \kappa_0^{(G=2)} (\r) &=& {7 \over 6}- \r -\frac{1}{6} \sqrt{1+12 \r -12 \r ^2} . \eea
\begin{figure}
	\centering
	\includegraphics[width=13cm]{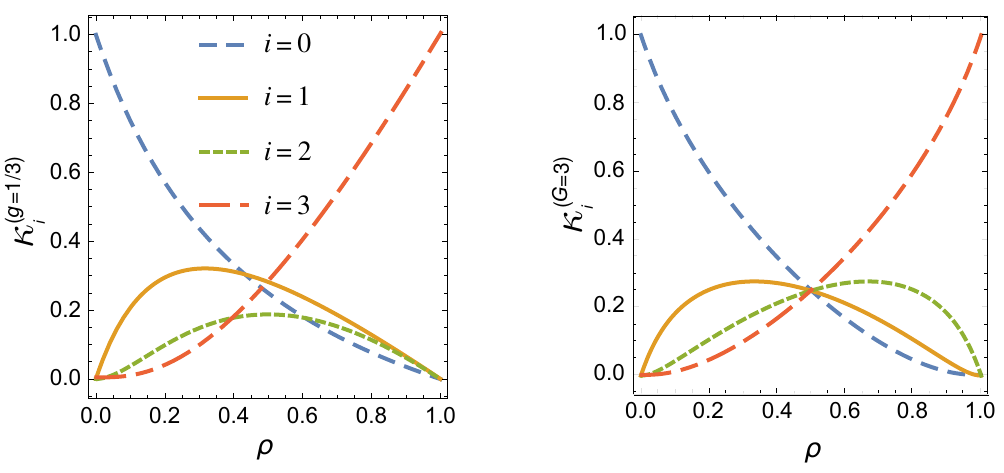}
	\caption{(Color online) The mean numbers $\kappa_0$, $\kappa_1$, $\kappa_2$, $\kappa_3$ vs $\r$, for FES$_{1/3}$ and GS$_3$.}\label{comp3} \end{figure}
Comparing $\kappa_i^{(\alpha)}$ for semion statistics and GS$_2$, we see in Fig.~\ref{kappa2} that qualitatively both models display a very close behavior. Also, the fluctuations of $k_2$ under both measures are comparable, although the distribution is slightly more dispersed under the Haldane one.
\begin{figure}
	\centering
	\includegraphics[width=13cm]{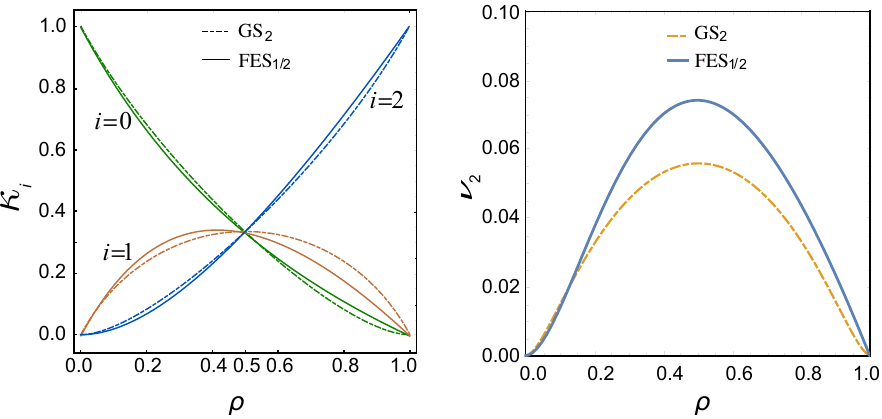}\\
	\caption{(Color online) The normalized mean numbers $\kappa_i$ ($i=0,1,2$) and The normalized variance $\nu_2$ vs. $\r$,  for semions and GS$_2$.} \label{kappa2}
\end{figure}
Particularly at half filling ($\r=1/2$), we highlight the following points: (i) For semions statistics and GS$_2$ all the $\kappa_i^{(\alpha)}$ coincide. This feature, though systematic for GS, is not shared by FES$_g$ with $r-q\geq 2$  (Fig.\ref{comp3}). (ii) It is easy to see that the unique positive solution of Eq.~\eqref{eq} for $\r=1/2$ is $x(1/2)=1$, and, consequently, $\kappa_i^{(G)} (1/2) = 1/(G+1)$ for all $i=0, \ldots G$, illustrating the uniform distribution of occupancies under Gentile measure at half-filling, (iii) If $g=1/G$, the mean numbers of occupied states are the same ($=G/(G+1)$) under the two measures (Fig.~\ref{ell3}). 

\section{Conclusion}\label{s4}
In the present work, we have shown that the generalized exclusion principle for FES cannot be fully understood without an exact combinatorics of~\eqref{HW}. The author believes however that this point deserve further elucidation, and any interpretation should shed more light on the subject. In fact, in the Haldane's seminal paper, the interpolating formula~\eqref{HW} was not derived from a concrete counting procedure as is the case for conventional statistics or GS. In Ref.~\cite{fahssi}, the author presented several interpretations of so-called polynomial coefficients, or extended binomial coefficients, given by~\eqref{G}, namely as number of restricted integer compositions, as score in drawing balls, as counting certain directed lattice paths or spin chain models, etc. It would be instructive to seek similar interpretations for $W_{g}(K,N)$. For example, in terms of generalized integer compositions, it has been observed~\cite[Sequence A078812]{OEIS} that $W_{g=1/2}(K,N)$ is the number of ways of writing $K$ as the sum of $P+1$ strictly positive integers when there are 1 kind of part 1, 2 kinds of part 2: $2_1$ and $2_2$, and so on~\footnote{See the comment of Emeric Deutsch on the sequence A078812~\cite{OEIS}.}. For example, $W_{1/2}(4,2) = 10$ since there are 10 such compositions of 4: $(1,3_1)$, $(1,3_2)$, $(1,3_3)$, $(3_1,1)$, $(3_2,1)$, $(3_3,1)$, $(2_1,2_1)$, $(2_1,2_2)$, $(2_2,2_1)$ and $(2_2,2_2)$. The main goal is to find a one-to-one correspondence between the set of such compositions and the set occupancy configurations for semions. 
\begin{figure}
	\centering
	\includegraphics[width=6cm]{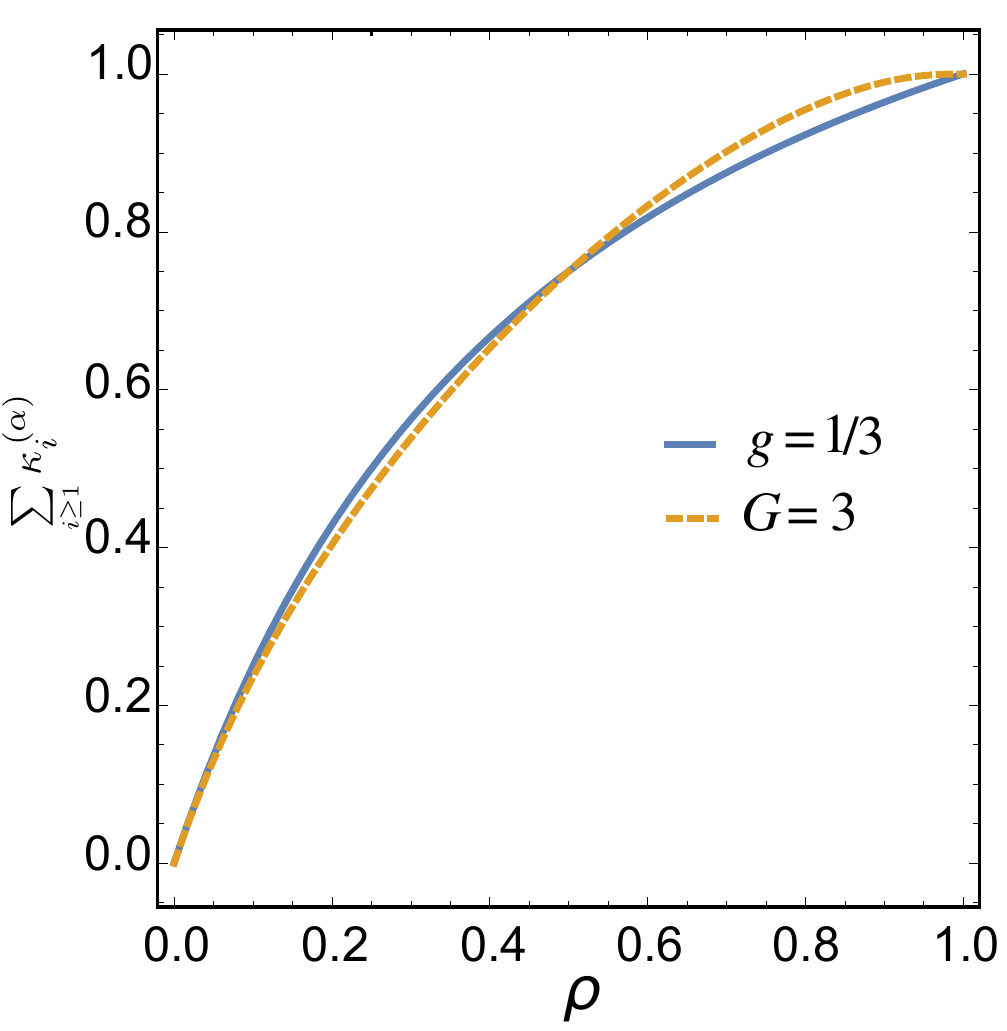}\\
	\caption{(Color online) The normalized mean number of occupied states vs. $\r$ in the thermodynamic limit,  for FES$_{1/3}$ and GS$_3$.} \label{ell3}
\end{figure}


\end{document}